\documentclass[11pt]{article} \usepackage{amsmath, amsthm, amssymb,
  dsfont}
\usepackage[pdftex,left=1in,top=1in,bottom=1in,right=1in]{geometry}

\title{Restricted Isometry of Fourier Matrices and \\ List Decodability of Random Linear
  Codes}

\author{{\sf Mahdi Cheraghchi}\thanks{%
Email: $\langle$cheraghchi@cmu.edu$\rangle$. %
Research supported in part by the Swiss National Science Foundation research grant PBELP2-133367.
} \quad\quad {\sf Venkatesan Guruswami}\thanks{%
Email: $\langle$guruswami@cmu.edu$\rangle$. 
Research supported in part by a Packard Fellowship and NSF CCF-0963975.
} \quad\quad
  {\sf Ameya Velingker}\thanks{%
Email: $\langle$avelingk@cs.cmu.edu$\rangle$. 
Research supported in part by NSF CCF-0963975.} \\ \\
  Computer Science Department \\
  Carnegie Mellon University\\
  Pittsburgh, PA 15213}
  
\date{}

\newcommand{\C}{\mathds{C}}
\newcommand{\cC}{\mathcal{C}}

\newcommand{\cE}{\mathcal{E}}

\newcommand{\F}{\mathds{F}}

\newcommand{\R}{\mathds{R}}
\newcommand{\E}{\mathds{E}}
\newcommand{\supp}{\mathsf{supp}}

\newcommand{\eps}{\epsilon}
\newcommand{\innr}[1]{\langle #1 \rangle}

\newtheorem{thm}{Theorem}
\newtheorem{prop}[thm]{Proposition}
\newtheorem{claim}[thm]{Claim}
\newtheorem{lem}[thm]{Lemma}

\newtheorem{observ}[thm]{Observation}
\theoremstyle{definition}

\newtheorem{defn}[thm]{Definition}

\newcommand{\lin}{\mathsf{Lin}}

\newcommand{\sgn}{\mathrm{sgn}}

\newcommand{\tk}{{\tilde{k}}}

\newcommand{\I}{\mathbf{i}}
\renewcommand{\vec}[1]{{#1}}

\newcommand{\Vnote}[1]{} % Resolved margin notes by Venkat
\newcommand{\Mnote}[1]{} % Resolved margin notes by Mahdi
\newcommand{\Anote}[1]{} % Resolved margin notes by Ameya

\newcommand{\Char}{\varphi}
\newcommand{\lsim}{\lesssim}
\newcommand{\gsim}{\gtrsim}
\newcommand{\erf}{\mathrm{erf}}
\newcommand{\erfc}{\mathrm{erfc}}

\newcommand{\Rows}{T}
\newcommand{\Ball}{\mathcal{B}_2^{k,N}}

\usepackage[hyphens]{url}
\usepackage[
     hidelinks, % comment out if you need link borders
     plainpages=false,
     pdftitle={Restricted Isometry of Fourier Matrices and List Decodability of Random Linear Codes},
     pdfsubject={},
     pdfstartview=FitH,
     bookmarks=true,
     colorlinks=false,
     bookmarksopen=false,
     bookmarksnumbered=true,
     pdfhighlight=/I,
     hyperfigures=true,
     pdfauthor={Mahdi Cheraghchi and Venkatesan Guruswami and Ameya Velingker}]{hyperref}
\usepackage{breakurl}

\begin{document}
     \maketitle %\thispagestyle{empty}

     \begin{abstract}
       We prove that a random linear code over $\F_q$, with
       probability arbitrarily close to $1$, is list decodable at
       radius $1-1/q-\eps$ with list size $L=O(1/\eps^2)$ and rate
       $R=\Omega_q(\eps^2/(\log^3(1/\eps)))$. Up to the
       polylogarithmic factor in $1/\eps$ and constant factors
       depending on $q$, this matches the lower bound
       $L=\Omega_q(1/\eps^2)$ for the list size and upper bound
       $R=O_q(\eps^2)$ for the rate. Previously only existence (and
       not abundance) of such codes was known for the special case
       $q=2$ (Guruswami, H{\aa}stad, Sudan and Zuckerman, 2002).

       In order to obtain our result, we employ a relaxed version of
       the well known Johnson bound on list decoding that translates
       the \emph{average} Hamming distance between codewords to list
       decoding guarantees. We furthermore prove that the desired
       average-distance guarantees hold for a code provided that a
       natural complex matrix encoding the codewords satisfies the
       Restricted Isometry Property with respect to the Euclidean norm
       (RIP-2). For the case of random binary linear codes, this
       matrix coincides with a random submatrix of the Hadamard-Walsh
       transform matrix that is well studied in the compressed sensing
       literature.

       Finally, we improve the analysis of Rudelson and Vershynin
       (2008) on the number of random frequency samples required for
       exact reconstruction of $k$-sparse signals of length
       $N$. Specifically, we improve the number of samples from $O(k
       \log(N) \log^2(k) (\log k + \log\log N))$ to $O(k \log(N) \cdot
       \log^3(k))$. The proof involves bounding the expected 
       supremum of a related Gaussian process by using an improved 
       analysis of the metric defined by the process. This improvement 
       is crucial for our application in list decoding.
     \end{abstract}

\newpage

%==============================================================================
%==============================================================================

\section{Introduction}

This work is motivated by the list decodability properties of random
linear codes for correcting a large fraction of errors, approaching
the information-theoretic maximum limit.  We prove a near-optimal
bound on the rate of such codes, by making a connection to and
establishing improved bounds on the restricted isometry property of
random submatrices of Hadamard matrices.

A $q$-ary error correcting code $\cC$ of block length $n$ is a subset
of $[q]^n$, where $[q]$ denotes any alphabet of size $q$. The rate of
such a code is defined to be $(\log_q |\cC|)/n$. A good code $\cC$
should be large (rate bounded away from $0$) and have its elements
(codewords) well ``spread out." The latter property is motivated by
the task of recovering a codeword $c \in \cC$ from a noisy version $r$
of it that differs from $c$ in a bounded number of coordinates. Since
a random string $r \in [q]^n$ will differ from $c$ on an expected
$(1-1/q)n$ positions, the information-theoretically maximum fraction
of errors one can correct is bounded by the limit $(1-1/q)$. In fact,
when the fraction of errors exceeds $\frac12 (1-1/q)$, it is not
possible to unambiguously identify the close-by codeword to the noisy
string $r$ (unless the code has very few codewords, i.e., a rate
approaching zero).

In the model of list decoding, however, recovery from a fraction of
errors approaching the limit $(1-1/q)$ becomes possible. Under list
decoding, the goal is to recover a small list of all codewords of
$\cC$ differing from an input string $r$ in at most $\rho n$
positions, where $\rho$ is the error fraction (our interest in this
paper being the case when $\rho$ is close to $1-1/q$). This requires
that $\cC$ have the following sparsity property, called {\em
  $(\rho,L)$-list decodability}, for some small $L$ : for every $r \in
[q]^n$, there are at most $L$ codewords within Hamming distance $\rho
n$ from $r$. We will refer to the parameter $L$ as the ``list size"
--- it refers to the maximum number of codewords that the decoder may
output when correcting a fraction $\rho$ of errors. Note that
$(\rho,L)$-list decodability is a strictly combinatorial notion, and
does not promise an efficient algorithm to compute the list of
close-by codewords. In this paper, we only focus on this combinatorial
aspect, and study a basic trade-off between between $\rho$, $L$, and
the rate for the important class of random linear codes, when $\rho
\to 1-1/q$. We describe the prior results in this direction and state
our results next.

For integers $q,L\ge 2$, a random $q$-ary code of rate $R = 1
-h_q(\rho)-1/L$ is $(\rho,L)$-list decodable with high
probability. Here $h_q\colon [0,1-1/q] \to [0,1]$ is the $q$-ary
entropy function: $h_q(x) = x \log_q(q-1) - x \log_q x - (1-x)
\log_q(1-x)$. This follows by a straightforward application of the
probabilistic method, based on a union bound over all centers $r \in
[q]^n$ and all $(L+1)$-element subsets $S$ of codewords that all
codewords in $S$ lie in the Hamming ball of radius $\rho n$ centered
at $r$. For $\rho = 1-1/q-\eps$, where we think of $q$ as fixed and
$\eps \to 0$, this implies that a random code of rate
$\Omega_q(\eps^2)$ is $(1-1/q-\eps,O_q(1/\eps^2))$-list
decodable. (Here and below, the notation $\Omega_q$ and $O_q$ hide
constant factors that depend only on $q$.)

Understanding list decodable codes at the extremal radii \Mnote{Added
  references on applications.}  $\rho = 1-1/q-\eps$, for small $\eps$,
is of particular significance mainly due to numerous applications that
depend on this regime of parameters. For example, one can mention
hardness amplification of Boolean functions \cite{ref:STV01},
construction of hardcore predicates from one-way functions
\cite{ref:GL89}, construction of pseudorandom generators
\cite{ref:STV01} and randomness extractors \cite{ref:Tre01},
inapproximability of $\mathsf{NP}$ witnesses \cite{ref:KS99}, and
approximating the VC dimension \cite{ref:MU01}. Moreover,
\emph{linear} list-decodable codes are further appealing due to their
symmetries, succinct description, and efficient encoding. For some
applications, linearity of list decodable codes is of crucial
importance. For example, the black-box reduction from list decodable
codes to capacity achieving codes for additive noise channels in
\cite{ref:GS10}, or certain applications of Trevisan's extractor
\cite{ref:Tre01} (e.g., \cite[\S~3.6, \S~5.2]{ref:Che10}) rely on
linearity of the underlying list decodable code. Furthermore, list
decoding of linear codes features an interplay between linear
subspaces and Hamming balls and their intersection properties, which
is of significant interest from a combinatorial perspective.

This work is focused on random {\em linear} codes, which are subspaces
of $\F_q^n$, where $\F_q$ is the finite field with $q$ elements.
% (henceforth therefore, we assume $q$ is a prime power).
A random linear code $\cC$ of rate $R$ is sampled by picking $k=Rn$
random vectors in $\F_q^n$ and letting $\cC$ be their
$\F_q$-span. Since the codewords of $\cC$ are now not all independent
(in fact they are not even $3$-wise independent), the above naive
argument only proves the $(\rho,L)$-list decodability property for
codes of rate $1-h_q(\rho)-1/\log_q (L+1)$~\cite{ZP81}.\footnote{The
  crux of the argument is that any $L$ non-zero vectors in $\F_q^k$
  must have a subset of $\log_q(L+1)$ linearly independent vectors,
  and these are mapped independently by a random linear code. This
  allows one to effectively substitute $\log_q (L+1)$ in the place of
  $L$ in the argument for fully random codes.} For the setting $\rho =
1-1/q-\eps$, this implies a list size bound of $\exp(O_q(1/\eps^2))$
for random linear codes of rate $\Omega_q(\eps^2)$, which is
exponentially worse than for random codes. Understanding if this
exponential discrepancy between general and linear codes is inherent
was raised an open question by Elias~\cite{elias91}. Despite much
research, the exponential bound was the best known for random linear
codes (except for the case of $q=2$, and even for $q=2$ only an
existence result was known; see the related results section below for
more details).

Our main result in this work closes this gap between random linear and
random codes, up to polylogarithmic factors in the rate. We state a
simplified version of the main theorem (Theorem~\ref{thm:main}) below.

\begin{thm}[Main, simplified]
  \label{thm:main-intro}
  Let $q$ be a prime power, and let $\eps > 0$ be a constant
  parameter. Then for some constant $a_q > 0$ only depending on $q$
  and all large enough integers $n$, a random linear code $\cC
  \subseteq \F_q^n$ of rate $a_q \eps^2/\log^3 (1/\eps)$ is
  $(1-1/q-\eps, O(1/\eps^2))$-list decodable with probability at least
  $0.99$. (one can take $a_q = \Omega(1/\log^4 q)$.)\Vnote{(DONE)
    Check $a_q$ and fix later.}
\end{thm}

We remark that both the rate and list size are close to optimal for
list decoding from a $(1-1/q-\eps)$ fraction of errors. For rate, this
follows from the fact the $q$-ary ``list decoding capacity" is given
by $1-h_q(\rho)$, which is $O_q(\eps^2)$ for $\rho=1-1/q-\eps$. For
list size, a lower bound of $\Omega_q(1/\eps^2)$ is known --- this
follows from \cite{blinovsky} for $q=2$, and was shown for all $q$ in
\cite{GV-ld-lb} (and also in \cite{blin-q-ary} under a convexity
conjecture that was later proved in \cite{blin-convexity}).  We have
also assumed that the alphabet size $q$ is fixed and have not
attempted to obtain the best possible dependence of the constants on
the alphabet size.

\subsection{Related results}

We now discuss some other previously known results concerning list
decodability of random linear codes.

First, it is well known that a random linear code of rate
$\Omega_q(\eps^4)$ is $(1-1/q-\eps,O(1/\eps^2))$-list decodable with
high probability. This follows by combining the Johnson bound for list
decoding (see, for example, \cite{GS-johnson}) with the fact that such
codes lie on the Gilbert-Varshamov bound and have relative distance
$1-1/q-\eps^2$ with high probability. This result gets the correct
quadratic dependence in list size, but the rate is worse.

Second, for the case of $q=2$, the existence of $(\rho,L)$-list
decodable binary linear codes of rate $1-h(\rho)-1/L$ was proved in
\cite{GHSZ}. For $\rho=1/2-\eps$, this implies the existence of binary
linear codes of rate $\Omega(\eps^2)$ list decodable with list size
$O(1/\eps^2)$ from an error fraction $1/2-\eps$. This matches the
bounds for random codes, and is optimal up to constant
factors. However, there are two shortcomings with this result: (i) it
only works for $q=2$ (the proof makes use of this in a crucial way,
and extensions of the proof to larger $q$ have been elusive), and (ii)
the proof is based on the semi-random method. It only shows the
existence of such a code while failing to give any sizeable lower
bound on the probability that a random linear code has the claimed
list decodability property.

Motivated by this state of affairs, in \cite{GHK}, the authors proved
that a random $q$-ary linear code of rate $1-h_q(\rho)-C_{\rho,q}/L$
is $(\rho,L)$-list decodable with high probability, for some
$C_{\rho,q} < \infty$ that depends on $\rho,q$. This matches the
result for completely random codes up to the leading constant
$C_{\rho,q}$ in front of $1/L$. Unfortunately, for $\rho =
1-1/q-\eps$, the constant $C_{\rho,q}$ depends exponentially\footnote{
  The constant $C_{\rho,q}$ depends exponentially on
  $1/\delta_{\rho}$, where $q^{-\delta_p n}$ is an upper bound on the
  probability that two random vectors in $\F_q^n$ of relative Hamming
  weight at most $\rho$, chosen independently and uniformly among all
  possibilities, sum up (over $\F_q^n$) to a vector of Hamming weight
  at most $\rho$. When $\rho = 1-1/q-\eps$, we have $\delta_\rho =
  \Theta_q(\eps^2)$ which makes the list size exponentially large.  }
on $1/\eps$. Thus, this result only implies an exponential list size
in $1/\eps$, as opposed to the optimal $O(1/\eps^2)$ that we
seek. \Vnote{Should say why/how $C_{1/2-\eps,1/2}$ depends on $\eps$?}

Summarizing, for random linear codes to achieve a polynomial in
$1/\eps$ list size bound for error fraction $1-1/q-\eps$, the best
lower bound on rate was $\Omega(\eps^4)$. We are able to show that
random linear codes achieve a list size growing quadratically in
$1/\eps$ for a rate of $\tilde{\Omega}(\eps^2)$. One downside of our
result is that we do not get a probability bound of $1-o(1)$, but only
$1-\gamma$ for any desired constant $\gamma > 0$ (essentially our rate
bound degrades by a $\log (1/\gamma)$ factor).

Finally, there are also some results showing limitations on list
decodability of random codes. It is known that both random codes and
random linear codes of rate $1-h_q(\rho)-\eta$ are, with high
probability, {\em not} $(\rho, c_{\rho,q}/\eta)$-list
decodable~\cite{rudra11,GN12}. For arbitrary (not necessarily random)
codes, the best lower bound on list size is $\Omega(\log
(1/\eta))$~\cite{blinovsky,GN12}.

\subsection{Proof technique}

The proof of our result uses a different approach from the earlier
works on list decodability of random linear codes
\cite{ZP81,elias91,GHSZ,GHK}. Our approach consists of three steps.

\medskip \noindent \textbf{Step 1:} Our starting point is a relaxed
version of the Johnson bound for list decoding that only requires the
\emph{average} pairwise distance of $L$ codewords to be large (where
$L$ is the target list size), instead of the minimum distance of the
code.

Technically, this extension is easy and pretty much follows by
inspecting the proof of the Johnson bound. This has recently been
observed for the binary case by Cheraghchi and is implicit in the
survey \cite{ref:Che11}. Here, we give a proof of the relaxed Johnson
bound for a more general setting of parameters, and apply it in a
setting where the usual Johnson bound is insufficient.  Furthermore,
as a side application, we show how the average version can be used to
bound the list decoding radius of codes which do not have too many
codewords close to any codeword --- such a bound was shown via a
different proof in \cite{GKZ08}, where it was used to establish the
list decodability of binary Reed-Muller codes up to their distance.

\medskip \noindent \textbf{Step 2:} Prove that the $L$-wise average
distance property of random linear codes is implied by the order $L$
restricted isometry property (RIP-2) of random submatrices of the
Hadamard matrix (or in general, matrices related to the Discrete
Fourier Transform).

This part is also easy technically, and our contribution lies in
making this connection between restricted isometry and list
decoding. The restricted isometry property has received much attention
lately due to its relevance to compressed sensing (cf.\
\cite{ref:Can08,ref:CRT06,ref:CRT06b,ref:CT06,ref:Don06}) \Vnote{DONE:
  Some citation to compressed sensing here?} and is also connected to
the Johnson-Lindenstrauss dimension reduction lemma~\cite{BDDW08}. Our
work shows another interesting application of this concept.

\medskip \noindent \textbf{Step 3:} Prove the needed restricted
isometry property of the matrix obtained by sampling rows of the
Hadamard matrix.

This is the most technical part of our proof. Let us focus on $q=2$
for simplicity, and let $H$ be the $N \times N$ Hadamard (Discrete
Fourier Transform) matrix with $N=2^n$, whose $(x,y)$'th entry is
$(-1)^{\innr{x,y}}$ for $x,y \in \{0,1\}^n$. We prove that (the scaled
version of) a random submatrix of $H$ formed by sampling a subset of
$m = O( k \log^3 k \log N)$ rows of $H$ satisfies RIP of order $k$
with probability $0.99$. This means that every $k$ columns of this
sampled matrix $M$ are nearly orthogonal --- formally, every $m \times
k$ submatrix of $M$ has all its $k$ singular values close to $1$.

For random matrices $m \times N$ with i.i.d Gaussian or $\pm 1$
entries, it is relatively easy to prove RIP-2 of order $k$ when $m =
O(k \log N)$~\cite{BDDW08}. Proving such a bound for submatrices of
the Discrete Fourier Transform (DFT) matrix (as conjectured in
\cite{ref:RV08}) has been an open problem for many years. \Vnote{Were
  RV the first to state this open question? Or is it Candes-Tao? Can
  we something stronger on why this is hard?}
The difficulty is that the entries within a row are no longer
independent, and not even triple-wise independent. The best proven
upper bound on $m$ for this case was $O(k \log^2 k (\log k + \log \log
N) \log N)$, improving an earlier upper bound $O(k \log^6 N)$ of
Cand\`es and Tao \cite{ref:CT06}.  We improve the bound to $O(k \log^3
k \log N)$ --- the key gain is that we do {\em not} have the $\log
\log N$ factor. This is crucial for our list decoding connection, as
the rate of the code associated with the matrix will be $(\log N)/m$,
which would be $o(1)$ if $m =\Omega(\log N \log \log N)$. We will take
$k = L = \Theta(1/\eps^2)$ (the target list size), and the rate of the
random linear code will be $\Omega(1/(k \log^3 k))$, giving the bounds
claimed in Theorem~\ref{thm:main-intro}. We remark that any
improvement of the RIP bound towards the information-theoretic limit
$m = \Omega(k \log(N/k))$, a challenging open problem, would
immediately translate into an improvement on the list decoding rate of
random linear codes via our reductions.

Our RIP-2 proof for row-subsampled DFT matrices proceeds along the
lines of \cite{ref:RV08}, and is based on upper bounding the
expectation of the supremum of a certain {\em Gaussian
  process}~\cite[Chap. 11]{ref:banach}. \Vnote{Is Chap. 11 the right
  one to cite? Mahdi: Yes.}  The index set of the Gaussian process is
$\Ball$, the set of all $k$-sparse unit vectors in $\R^N$, and the
Gaussian random variable $G_x$ associated with $x\in \Ball$ is a
Gaussian linear combination of the squared projections of $x$ on the
rows sampled from the DFT matrix (in the binary case these are just
squared Fourier coefficients)\footnote{We should remark that our setup
  of the Gaussian process is slightly different from \cite{ref:RV08},
  where the index set is $k$-element subsets of $[N]$, and the
  associated Gaussian random variable is the spectral norm of a random
  matrix. Moreover, in \cite{ref:RV08} the number of rows of the
  subsampled DFT matrix is a random variable concentrating around its
  expectation, contrary to our case where it is a fixed number.  We
  believe that the former difference in our setup may make the proof
  accessible to a broader audience.}.  The key to analyzing the
Gaussian process is an understanding of the associated (pseudo)-metric
$X$ on the index set, defined by $\|x-x'\|_X^2 = \E_G
|G_x-G_{x'}|^2$. This metric is difficult to work with directly, so we
upper bound distances under $X$ in terms of distances under a
different metric $X'$. The principal difference in our analysis
compared to \cite{ref:RV08} is in the choice of $X'$ --- instead of
the max norm used in \cite{ref:RV08}, we use a large finite norm
applied to the sampled Fourier coefficients. We then estimate the
covering numbers for $X'$ and use Dudley's theorem to bound the
supremum of the Gaussian process.

It is worth pointing out that, as we prove in this work,
for low-rate random linear codes the average-distance quantity
discussed in Step~1 above is substantially larger than the 
minimum distance of the code. This allows the relaxed version
of the Johnson bound attain better bounds than what the standard
(minimum-distance based) Johnson bound would obtain on 
list decodability of random linear codes. While explicit examples
of linear codes surpassing the standard Johnson bound are
already known in the literature (see \cite{GGR11} and
the references therein), a by-product of our result is that in fact
\emph{most} linear codes (at least in the low-rate regime)
surpass the standard Johnson bound. However, an interesting 
question is to see whether there are codes that are list decodable
even beyond the relaxed version of the Johnson bound studied 
in this work.

\medskip \noindent \textbf{Organization of the paper. }
The rest of the paper is organized as follows. After fixing some
notation, in Section~\ref{sec:Johnson} we prove the average-case
Johnson bound that relates a lower bound on average pair-wise
distances of subsets of codewords in a code to list decoding
guarantees on the code. We also show, in 
Section~\ref{sec:locallySparse}, an application of this bound 
on proving list decodability of ``locally sparse'' codes,
which is of independent interest and simplifies some earlier 
list decoding results. 
In Section~\ref{sec:listDecoding}, we prove our
main theorem on list decodability of random linear codes by
demonstrating a reduction from RIP-2 guarantees of DFT-based complex
matrices to average distance of random linear codes, combined with the
Johnson bound. Finally, the RIP-2 bounds on matrices related to random
linear codes are proved in Section~\ref{sec:RIP}.

\medskip \noindent \textbf{Notation. }
Throughout the paper, we will be interested in list decodability of
$q$-ary codes. We will denote an alphabet of size $q$ by $[q]$ (which
one can identify with $\{0,1,\dots,q-1\}$); for linear codes, the
alphabet will be $\F_q$, the finite field with $q$ elements (when $q$
is a prime power).

We use the notation $\I := \sqrt{-1}$.  When $f \leq C g$ (resp., $f
\geq C g$) for some absolute constant $C > 0$, we use the shorthand $f
\lsim g$ (resp., $f \gsim g$).  We use the notation $\log(\cdot)$ when
the base of logarithm is not of significance (e.g., $f \lsim \log
x$). Otherwise the base is subscripted as in $\log_b(x)$. The natural
logarithm is denoted by $\ln(\cdot)$.

For a matrix $M$ and a multiset of rows $\Rows$, define $M_\Rows$ to
be the matrix with $|\Rows|$ rows, formed by the rows of $M$ picked by
$\Rows$ (in some arbitrary order).  Each row in $M_\Rows$ may be
repeated for the appropriate number of times specified by $\Rows$.

% ==============================================================================
% ==============================================================================

\section{Average-distance based Johnson bound} \label{sec:Johnson}

In this section, we show how the average pair-wise distances between
subsets of codewords in a $q$-ary code translate into list
decodability guarantees on the code.

% \begin{defn}
Recall that the relative Hamming distance between strings $x,y \in
[q]^n$, denoted $\delta(x,y)$, is defined to be the fraction of
positions $i$ for which $x_i \neq y_i$.  The relative distance of a
code $\cC$ is the minimum value of $\delta(x,y)$ over all pairs of
codewords $x \neq y \in \cC$. We define list decodability as follows.
% \end{defn}

\begin{defn}
  A code $\cC \subseteq [q]^n$ is said to be $(\rho, \ell)$-list
  decodable if $\forall y \in [q]^n$, the number of codewords of $\cC$
  within relative Hamming distance less than $\rho$ is at most
  $\ell$.\footnote{We require that the radius is strictly less than
    $\rho$ instead of at most $\rho$ for convenience.}
\end{defn}

The following definition captures a crucial function that allows one
to generically pass from distance property to list decodability.

\begin{defn}[Johnson radius]
  For an integer $q \ge 2$, the Johnson radius function $J_q :
  [0,1-1/q] \rightarrow [0,1]$ is defined by
  \[ J_q(x) := \frac{q-1}{q} \left( 1 - \sqrt{ 1 - \frac{q x}{q-1} }
  \right) \ . \]
\end{defn}

The well known Johnson bound in coding theory states that a $q$-ary
code of relative distance $\delta$ is $(J_q(\delta-\delta/L), L)$-list
decodable (see for instance \cite{GS-johnson}). Below we prove a
version of this bound which does not need every pair of codewords to
be far apart but instead works when the average distance of a set of
codewords is large. The proof of this version of the Johnson bound is
a simple modification of earlier proofs, but working with this version
is a crucial step in our near-tight analysis of the list decodability
of random linear codes.

\begin{thm}[Average-distance Johnson bound]
  \label{thm:avg-johnson}
  Let $\cC \subseteq [q]^n$ be a $q$-ary code and $L \ge 2$ an
  integer. If the average pairwise relative Hamming distance of every
  subset of $L$ codewords of $\cC$ is at least $\delta$, then $\cC$ is
  $(J_q(\delta-\delta/L), L-1)$-list decodable.
\end{thm}

Thus, if one is interested in a bound for list decoding with list size
$L$, it is enough to consider the average pairwise Hamming distance of
subsets of $L$ codewords.

\subsection{Geometric encoding of $q$-ary symbols}
We will give a geometric proof of the above result. For this purpose,
we will map vectors in $[q]^n$ to complex vectors and argue about the
inner products of the resulting vectors.

\begin{defn}[Simplex encoding] \label{def:simplex} The simplex
  encoding maps $x \in \{0,1,\cdots,q-1\}$ to a vector $\Char(x) \in
  \C^{q-1}$.  The coordinate positions of this vector are indexed by
  the elements of $[q-1] := \{1,2,\dots,q-1\}$.  Namely, for every
  $\alpha \in [q-1]$, we define $\Char(x)(\alpha) := \omega^{x\alpha}$
  where $\omega = e^{2\pi \I/q}$ is the primitive $q$th complex root
  of unity.
\end{defn}

For complex vectors $\vec{v}=(v_1,v_2,\dots,v_m)$ and $\vec{w}
=(w_1,w_2,\dots,w_m)$, we define their inner product
$\innr{\vec{v},\vec{w}} = \sum_{i=1}^m v_i w_i^*$. From the definition
of the simplex encoding, the following immediately follows:
\begin{equation*} %\label{eqn:simplex}
  \innr{\Char(x), \Char(y)} = \left\{ \begin{array}{ll}
      q-1 & \text{if $x=y$}, \\
      -1 & \text{if $x\neq y$}.
    \end{array} \right.
\end{equation*}
We can extend the above encoding to map elements of $[q]^n$ into
$\C^{n(q-1)}$ in the natural way by applying this encoding to each
coordinate separately. From the above inner product formula, it
follows that for $x, y \in [q]^n$ we have
\begin{equation}
  \label{eq:simplex-enc-dist}
  \innr{\Char(x), \Char(y)} = (q-1)n -q\delta(x,y) n  \ .
\end{equation}
Similarly, we overload the notation to matrices with entries over
$[q]$. Let $M$ be a matrix in $[q]^{n \times N}$. Then, $\Char(M)$ is
an $n(q-1) \times N$ complex matrix obtained from $M$ by replacing
each entry with its simplex encoding, considered as a column complex
vector.

Finally, we extend the encoding to \emph{sets} of vectors (i.e., codes) as well.
For a set $\cC \subseteq [q]^n$, $\Char(\cC)$ is defined as a $(q-1)n
\times |\cC|$ matrix with columns indexed by the elements of $\cC$,
where the column corresponding to each $c \in \cC$ is set to be
$\Char(c)$.

\subsection{Proof of average-distance Johnson bound}

We now prove the Johnson bound based on average distance.

\begin{proof}[Proof (of Theorem~\ref{thm:avg-johnson})] Suppose
  $\{c_1,c_2,\dots,c_L\} \subseteq [q]^n$ are such that their average
  pairwise relative distance is at least $\delta$, i.e.,
  \begin{equation}
    \label{eq:avg-dist-cond}
    \sum_{1 \le i < j \le L} \delta(c_i,c_j)  \ge \delta\cdot  {L \choose 2} \ .
  \end{equation}
  We will prove that $c_1,c_2,\dots,c_L$ cannot all lie in a Hamming
  ball of radius less than $J_q(\delta - \delta/L) n$.  Since every
  subset of $L$ codewords of $\cC$ satisfy \eqref{eq:avg-dist-cond},
  this will prove that $\cC$ is $(J_q(\delta - \delta/L), L-1)$-list
  decodable.

  Suppose, for contradiction, that there exists $c_0 \in [q]^n$ such
  that $\delta(c_0,c_i) \le \rho$ for $i=1,2,\dots,L$ and some $\rho <
  J_q(\delta - \delta/L)$. Recalling the definition of $J_q(\cdot)$,
  note that the assumption about $\rho$ implies
  \begin{equation}
    \label{eq:tau}
    \left( 1 -\frac{q \rho}{q-1} \right)^2 > 1  -\frac{q\delta }{q-1} + \frac{q}{q-1} \frac{\delta}{L} \ . 
  \end{equation}

  For $i=1,2,\dots,L$, define the vector $v_i = \Char(c_i) - \beta
  \Char(c_0) \in \C^{n(q-1)}$, for some parameter $\beta$ to be chosen
  later.  By \eqref{eq:simplex-enc-dist} and the assumptions about
  $c_0,c_1,\dots,c_L$, we have $\innr{\Char(c_i),\Char(c_0)} \ge (q-1)
  n - q \rho n$, and $\sum_{1 \le i < j \le L}
  \innr{\Char(c_i),\Char(c_j)} \le {L \choose 2} \bigl( (q-1)n - q
  \delta n \bigr)$. We have
  \begin{align*}
    0 & \le   \bigg\langle \sum_{i=1}^L v_i, ~\sum_{i=1}^L v_i \bigg\rangle = \sum_{i=1}^L \innr{v_i,v_i} + 2 \cdot \sum_{1\le i < j \le L} \innr{v_i,v_j} \\
    &  \leq L \bigl( n(q-1) + \beta^2 n (q-1) - 2 \beta (n(q-1) - q \rho n )\bigr) + \\
    & \quad +  L(L-1)  \bigl(  n(q-1) - q \delta n + \beta^2 n (q-1) - 2 \beta (n(q-1) - q \rho n ) \bigr) \\
    & = L^2 n (q-1) \left( \frac{q}{q-1} \frac{\delta}{L} + \biggl( 1-
      \frac{q \delta}{q-1} + \beta^2 - 2 \beta \Bigl( 1 - \frac{q
        \rho}{q-1} \Bigr) \biggr) \right)
  \end{align*}
  Picking $\beta = 1 - \frac{q \rho}{q-1}$ and recalling
  \eqref{eq:tau}, we see that the above expression is negative, a
  contradiction.
\end{proof}

\subsection{An application: List decodability of Reed-Muller and
  locally sparse codes} \label{sec:locallySparse}

Our average-distance Johnson bound implies the following combinatorial
result for the list decodability of codes that have few codewords in a
certain vicinity of every codeword. The result allows one to translate
a bound on the number of codewords in balls centered at codewords to a
bound on the number of codewords in an arbitrary Hamming ball of
smaller radius.  An alternate proof of the below bound (using a
``deletion" technique) was given by Gopalan, Klivans, and
Zuckerman~\cite{GKZ08} where they used it to argue the list
decodability of (binary) Reed-Muller codes up to their relative
distance. A mild strengthening of the deletion lemma was later used in
\cite{GGR11} to prove combinatorial bounds on the list decodability of
tensor products and interleavings of binary linear codes.

\begin{lem}
  \label{lem:deletion}
  Let $q \ge 2$ be an integer and $\eta \in (0,1-1/q]$. Suppose $\cC$
  is a $q$-ary code such that for every $c \in \cC$, there are at most
  $A$ codewords of relative distance less than $\eta$ from $c$
  (including $c$ itself).  Then, for every positive integer $L \ge 2$,
  $\cC$ is $(J_q(\eta-\eta/L), A L-1)$-list decodable.
\end{lem}

Note that setting $A=1$ above gives the usual Johnson bound for a code
of relative distance at least $\eta$.

\begin{proof}
  We will lower bound the average pairwise relative distance of every
  subset of $AL$ codewords of $\cC$, and then apply
  Theorem~\ref{thm:avg-johnson}.

  Let $c_1,c_2,\dots,c_{AL}$ be distinct codewords of $\cC$. For
  $i=1,2,\dots,AL$, the sum of relative distances of $c_j$, $j\neq i$,
  from $c_i$ is at least $(AL-A) \eta$ since there are at most $A$
  codewords at relative distance less than $\eta$ from
  $c_i$. Therefore
  \[ \frac{1}{{{AL} \choose 2}} \cdot \sum_{1 \le i < j \le AL}
  \delta(c_i,c_j) \ge \frac{AL \cdot (AL-A) \eta}{AL(AL-1)} =
  \frac{A(L-1)}{AL-1} \eta \ . \] Setting $\eta' = \frac{A (L-1)
    \eta}{AL-1}$, Theorem~\ref{thm:avg-johnson} implies that $\cC$ is
  $(J_q(\eta'-\frac{\eta'}{AL}), A L-1)$-list decodable. But $\eta' -
  \frac{\eta'}{AL} = \eta - \eta/L$, so the claim follows.
\end{proof}

% ==============================================================================
% ==============================================================================

\section{Proof of the list decoding result} \label{sec:listDecoding}

In this section, we prove our main result on list decodability of
random linear codes. The main idea is to use the \emph{restricted
  isometry property (RIP)} of complex matrices arising from random
linear codes for bounding average pairwise distances of subsets of
codewords. Combined with the average-distance based Johnson bound
shown in the previous section, this proves the desired list decoding
bounds. The RIP-2 condition that we use in this work is defined as
follows.

\begin{defn} \label{def:RIP} We say that a complex matrix $M \in \C^{m
    \times N}$ satisfies RIP-2 of order $k$ with constant $\delta$ if,
  for any $k$-sparse vector $x \in \C^N$, we have\footnote{We stress
    that in this work, we crucially use the fact that the definition
    of RIP that we use is based on the Euclidean ($\ell_2$) norm.}
  \[
  (1-\delta) \| x \|_2^2 \leq \| Mx \|_2^2 \leq (1+\delta) \| x
  \|_2^2.
  \]
  Generally we think of $\delta$ as a small positive constant, say
  $\delta = 1/2$.
\end{defn}

Since we will be working with list decoding radii close to $1-1/q$, we
derive a simplified expression for the Johnson bound in this regime;
namely, the following:

\begin{thm}
  \label{thm:avg-johnson-simple}
  Let $\cC \subseteq [q]^n$ be a $q$-ary code and $L \ge 2$ an
  integer. If the average pairwise relative Hamming distance of every
  subset of $L$ codewords of $\cC$ is at least $(1-1/q)(1-\eps)$, then
  $\cC$ is $((1-1/q)(1-\sqrt{\eps+1/L}), L-1)$-list decodable.
\end{thm}

\begin{proof}
  The proof is nothing but a simple manipulation of the bound given by
  Theorem~\ref{thm:avg-johnson}. Let $\delta := (1-1/q)(1-\eps)$.
  Theorem~\ref{thm:avg-johnson} implies that $\cC$ is
  $(J_q(\delta(1-1/L)), L-1)$-list decodable. Now,
  \begin{align*}
    J_q(\delta(1-1/L)) &= \frac{q-1}{q} \left( 1 - \sqrt{ 1 -
        \frac{q}{q-1} \cdot \frac{q-1}{q}\big(1-\eps
        \big)\Big(1-\frac{1}{L} \Big) } \right) \\ %
    &= \frac{q-1}{q} \left( 1 - \sqrt{ \eps+\frac{1}{L}-\frac{\eps}{L}
      } \right) \geq \frac{q-1}{q} \left( 1 - \sqrt{ \eps+\frac{1}{L}
      } \right). %
    \tag*{\qedhere}
  \end{align*}
\end{proof}

In order to prove lower bounds on average distance of random linear
codes, we will use the simplex encoding of vectors
(Definition~\ref{def:simplex}), along with the following simple
geometric lemma.

\begin{lem} \label{lem:avgDist} Let $c_1, \ldots, c_L \in [q]^n$ be
  $q$-ary vectors. Then, the average pairwise distance $\delta$
  between these vectors satisfies
  \[
  \delta := \sum_{1 \leq i<j \leq L} \delta(c_i, c_j)/\binom{L}{2} =
  \frac{L^2(q-1)n - \Big\| \sum_{i \in [L]} \Char(c_i) \Big\|_2^2}{q
    L(L-1) n}.
  \]
\end{lem}

\begin{proof}
  The proof is a simple application of
  \eqref{eq:simplex-enc-dist}. The second norm on the right hand side
  can be expanded as
  \begin{eqnarray*}
    \Big\| \sum_{i \in [L]} \Char(c_i) \Big\|_2^2 &=& \sum_{i, j \in [L]} \innr{\Char(c_i), \Char(c_j)} \\
    &\stackrel{\eqref{eq:simplex-enc-dist}}{=}& \sum_{i, j \in [L]} \Big((q-1)n -qn\delta(c_i,c_j) \Big) \\
    &=& L^2 (q-1)n - 2qn \sum_{1 \leq i < j \leq L} \delta(c_i,c_j) \\
    &=& L^2 (q-1)n - 2qn \binom{L}{2} \delta,
  \end{eqnarray*}
  and the bound follows.
\end{proof}

Now we are ready to formulate our reduction from RIP-2 to average
distance of codes.

\begin{lem} \label{lem:RIPtoDist} Let $\cC \subseteq [q]^n$ be a code
  and suppose $\Char(\cC)/\sqrt{(q-1)n}$ satisfies RIP-2 of order $L$
  with constant $1/2$. Then, the average pairwise distance between
  every $L$ codewords of $\cC$ is at least $
  \big(1-\frac{1}{q}\big)\big(1 - \frac{1}{2(L-1)} \big)$.
\end{lem}

\begin{proof} 
  Consider any set $S$ of $L$ codewords, and the real vector $x \in
  \R^{|\cC|}$ with entries in $\{0,1\}$ that is exactly supported on
  the positions indexed by the codewords in $S$.  Obviously,
  $\|x\|_2^2 = L$. Thus, by the definition of RIP-2
  (Definition~\ref{def:RIP}), we know that, defining $M :=
  \Char(\cC)$,
  \begin{equation} \label{eqn:lem:RIPtoDist} \| Mx \|_2^2 \leq
    3L(q-1)n/2.
  \end{equation}
  Observe that $Mx = \sum_{i \in [L]} \Char(c_i)$.  Let $\delta$ be
  the average pairwise distance between codewords in $S$.  By
  Lemma~\ref{lem:avgDist} we conclude that
  \begin{align*}
    \delta &=
    \frac{L^2(q-1)n - \Big\| \sum_{i \in [L]} \Char(c_i) \Big\|_2^2}{2q\binom{L}{2} n} \\
    &\stackrel{\eqref{eqn:lem:RIPtoDist}}{\geq}
    \frac{(L^2 - 1.5 L)(q-1)n}{q L(L-1) n} \\
    &= \frac{q-1}{q} \Big( 1 - \frac{1}{2(L-1)}\Big).  %
    \tag*{\qedhere}
  \end{align*}
\end{proof}

We remark that the exact choice of the RIP constant in the above
result is arbitrary, as long as it remains an absolute
constant. Contrary to applications in compressed sensing, for our
application it also makes sense to have RIP-2 with constants larger
than one, since the proof only requires the upper bound in
Definition~\ref{def:RIP}.

By combining Lemma~\ref{lem:RIPtoDist} with the simplified Johnson
bound of Theorem~\ref{thm:avg-johnson-simple}, we obtain the following
corollary.

\begin{thm} \label{thm:RIPtoLD} Let $\cC \subseteq [q]^n$ be a code
  and suppose $\Char(\cC)/\sqrt{(q-1)n}$ satisfies RIP-2 of order $L$
  with constant $1/2$. Then $\cC$ is
  $\Big(\big(1-\frac{1}{q}\big)\big(1-\sqrt{\frac{1.5}{L-1}}\big),
  L-1\Big)$-list decodable. \qed
\end{thm}

The matrix $\Char(\cC)$ for a linear code $\cC \subseteq \F_q^n$ has a
special form. It is straightforward to observe that, when $q=2$, the
matrix is an incomplete Hadamard-Walsh transform matrix with $2^\tk$
columns, where $\tk$ is the dimension of the code. In general
$\Char(\cC)$ turns out to be related to a Discrete Fourier Transform
matrix. Specifically, we have the following observation.

\begin{observ} \label{obs:DFT} Let $\cC \subseteq \F_q^n$ be an $[n,
  \tk]$ linear code with a generator matrix $G \in \F_q^{\tk \times
    n}$, and define $N := q^\tk$.  Consider the matrix of \emph{linear
    forms} $\lin \in \F_q^{N \times N}$ with rows and columns indexed
  by elements of $\F_q^\tk$ and entries defined by
  \[
  \lin(x,y) := \innr{x,y},
  \]
  where $\innr{\cdot, \cdot}$ is the finite-field inner product over
  $\F_q^\tk$.  Let $\Rows \subseteq \F_q^\tk$ be the multiset of
  columns of $G$.  Then, $\Char(\cC) = \Char(\lin_{\Rows})$ (recall,
  from Definition~\ref{def:simplex}, that the former simplex encoding
  is applied to the matrix enumerating the codewords of $\cC$, while
  the latter is applied to the entries of a submatrix of $\lin$).
  
  When $G$ is uniformly random, $\cC$ becomes a random linear code and
  $\Char(\cC)$ can be sampled by the following process: Arrange $n$
  uniformly random rows of $\lin$, sampled independently and with
  replacement, as rows of a matrix $M$.  Then, replace each entry of
  $M$ by its simplex encoding, seen as a column vector in
  $\C^{q-1}$. The resulting complex matrix is $\Char(\cC)$.
\end{observ}

The RIP-2 condition for random complex matrices arising from random
linear codes is proved in Theorem~\ref{thm:RIP} of
Section~\ref{sec:RIP}.  We now combine this theorem with the preceding
results of this section to prove our main theorem on list decodability
of random linear codes.

\begin{thm}[Main] \label{thm:main} Let $q$ be a prime power, and let
  $\eps, \gamma > 0$ be constant parameters.  Then for all large
  enough integers $n$, a random linear code $\cC \subseteq \F_q^n$ of
  rate $R$, for some
  \[R \gsim \frac{\eps^2}{\log(1/\gamma) \log^3(q/\eps) \log q}\] is
  $((1-1/q)(1-\eps), O(1/\eps^2))$-list decodable with probability at
  least $1-\gamma$.
\end{thm}

\begin{proof}
  Let $\cC \subseteq \F_q^n$ be a uniformly random linear code
  associated to a random $Rn \times n$ generator matrix $G$ over
  $\F_q$, for a rate parameter $R \leq 1$ to be determined
  later. Consider the random matrix $M = \Char(\cC) =
  \Char(\lin_\Rows)$ from Observation~\ref{obs:DFT}, where $|\Rows| =
  n$. Recall that $M$ is a $(q-1)n \times N$ complex matrix, where $N
  = q^{Rn}$.  Let $L := 1 + \lceil 1.5/\eps^2 \rceil =
  \Theta(1/\eps^2)$.  By Theorem~\ref{thm:RIP}, for large enough $N$
  (thus, large enough $n$) and with probability $1-\gamma$, the matrix
  $M/\sqrt{(q-1)n}$ satisfies RIP-2 of order $L$ with constant $1/2$,
  for some choice of $|\Rows|$ bounded by
  \begin{equation} \label{eqn:main-bound} n = |\Rows| \lsim
    \log(1/\gamma) L \log(N) \log^3 (qL).
  \end{equation} 
  Suppose $n$ is large enough and satisfies \eqref{eqn:main-bound} so
  that the RIP-2 condition holds.  By Theorem~\ref{thm:RIPtoLD}, this
  ensures that the code $\cC$ is $((1-1/q)(1-\eps), O(1/\eps^2))$-list
  decodable with probability at least $1-\gamma$.
  
  It remains to verify the bound on the rate of $\cC$.  We observe
  that, whenever the RIP-2 condition is satisfied, $G$ must have rank
  exactly $Rn$, since otherwise, there would be distinct vectors $x,
  x' \in \F_q^{Rn}$ such that $xG=x'G$. Thus in that case, the columns
  of $M$ corresponding to $x$ and $x'$ become identical, implying that
  $M$ cannot satisfy RIP-2 of any nontrivial order.  Thus we can
  assume that the rate of $\cC$ is indeed equal to $R$.  Now we have
  \begin{eqnarray*}
    R &=& \log_q |\cC|/n = \log N / (n \log q) \\
    &\stackrel{\eqref{eqn:main-bound}}{\gsim}& \frac{\log N}{\log(1/\gamma) L \log(N) \log^3 (qL) \log q}.
  \end{eqnarray*}
  Substituting $L=\Theta(1/\eps^2)$ into the above expression yields
  the desired bound.
\end{proof}

% ==============================================================================
% ==============================================================================

\section{Restricted isometry property of DFT-based
  matrices} \label{sec:RIP}

In this section, we prove RIP-2 for random incomplete Discrete Fourier
Transform matrices. Namely, we prove the following theorem.

\begin{thm} \label{thm:RIP} Let $\Rows$ be a random multiset of rows
  of $\lin$, where $|\Rows|$ is fixed and each element of $\Rows$ is
  chosen uniformly at random, and independently with replacement.
  Then, for every $\delta, \gamma > 0$, and assuming $N \geq
  N_0(\delta, \gamma)$, with probability at least $1-\gamma$ the
  matrix $\Char(\lin_\Rows)/\sqrt{(q-1)|\Rows|}$ (with $(q-1)|\Rows|$
  rows) satisfies RIP-2 of order $k$ with constant $\delta$ for a
  choice of $|\Rows|$ satisfying
  \begin{equation} \label{eqn:RIP:rows} |\Rows| \lsim
    \frac{\log(1/\gamma)}{\delta^2} k \log(N) \log^3(qk).
  \end{equation}
\end{thm}

The proof extends and closely follows the original proof of Rudelson
and Vershynin \cite{ref:RV08}. However we modify the proof at a
crucial point to obtain a strict improvement over their original
analysis which is necessary for our list decoding application. We
present our improved analysis in this section.

\begin{proof}[Proof (of Theorem~\ref{thm:RIP})]
  Let $M := \Char(\lin_\Rows)$.  Each row of $M$ is indexed by an
  element of $\Rows$ and some $\alpha \in \F_q^\ast$ (recall that
  $\Rows \subseteq \F_q^\tk$, where $N = q^\tk$).  Denote the row
  corresponding to $t \in \Rows$ and $\alpha \in \F_q^\ast$ by $M_{t,
    \alpha}$, and moreover, denote the set of $k$-sparse unit
  vectors in $\C^N$ by $\Ball$.

  In order to show that $M/\sqrt{(q-1)|\Rows|}$ satisfies RIP of order
  $k$, we need to verify that for any $x = (x_1, \ldots, x_N) \in
  \Ball$,
  \begin{equation} \label{eqn:RIP} |\Rows|(q-1)(1-\delta) \leq \| Mx
    \|_2^2 \leq |\Rows|(q-1)(1+\delta).
  \end{equation}
  In light of Proposition~\ref{prop:RIPreal}, without loss of
  generality we can assume that $x$ is real-valued (since the inner
  product between any pair of columns of $M$ is real-valued).

  For $i \in \F_q^n$, denote the $i$th column of $M$ by $M^i$.  For $x
  = (x_1, \ldots, x_N) \in \Ball$, define the random variable
  \begin{eqnarray*}
    \Delta_x &:=& \| Mx \|_2^2 - |\Rows|(q-1) \\
    &=& \sum_{\substack{i, j \in \supp(x) \\ i \neq j}} x_i x_j \innr{M^i, M^j},
  \end{eqnarray*}
  where the second equality holds since each column of $M$ has
  $\ell_2$ norm $\sqrt{(q-1)|\Rows|}$ and $\|x\|_2 = 1$.  Thus, the
  RIP-condition \eqref{eqn:RIP} is equivalent to
  \begin{equation} \label{eqn:RIP:equiv} \Delta := \sup_{x \in \Ball}
    |\Delta_x| \leq \delta |\Rows|(q-1).
  \end{equation}
  Recall that $\Delta$ is a random variable depending on the
  randomness in $\Rows$. The proof of the RIP condition involves two
  steps. First, bounding $\Delta$ in expectation, and second, a tail
  bound.  The first step is proved, in detail, in the following lemma.
  
  \begin{lem} \label{lem:RIP:expectation} Let $\delta' > 0$ be a real
    parameter. Then, $\E[\Delta] \leq \delta' |\Rows| (q-1)$ for a
    choice of $|\Rows|$ bounded as follows:
    \begin{equation*}
      |\Rows| \lsim k \log(N) \log^3(qk) / {\delta'}^2.  
    \end{equation*}
  \end{lem}
  \begin{proof}
    We begin by observing that the columns of $M$ are orthogonal in
    expectation; i.e., for any $i, j \in \F_q^n$, we have
    \[
    \E_\Rows \innr{M^i, M^j} = \left\{ \begin{array}{ll}
        |\Rows| (q-1)  & i = j, \\
        0 & i \neq j.
      \end{array} \right.
    \]
    This follows from \eqref{eq:simplex-enc-dist} and the fact that
    the expected relative Hamming distance between the columns of
    $\lin$ corresponding to $i$ and $j$, when $i \neq j$, is exactly
    $1-1/q$.  It follows that for every $x \in \Ball$, $\E[\Delta_x] =
    0$, namely, the stochastic process $\{ \Delta_x \}_{x \in \Ball}$
    is centered.
  
    Recall that we wish to estimate
    \begin{eqnarray}
      \cE &:=& \E_\Rows \Delta \nonumber \\
      &=& \E_\Rows \sup_{x \in \Ball} \left|\sum_{t \in \Rows} \sum_{\alpha \in \F_q^*} \innr{M_{t,\alpha}, x}^2
        - |\Rows|(q-1)\right|. \label{eqn:RIP:w}
    \end{eqnarray}
    The random variables $\innr{M_{t,\alpha}, x}$ and
    $\innr{M_{t',\alpha'},x}$ are independent whenever $t \neq
    t'$. Therefore, we can use the standard symmetrization technique
    on summation of independent random variables in a stochastic
    process (Proposition~\ref{prop:sym}) and conclude from
    \eqref{eqn:RIP:w} that
    \begin{equation} \label{eqn:RIP:c} \cE \lsim \cE_1 := \E_\Rows
      \E_{\mathcal{G}} \sup_{x \in \Ball} \left(\sum_{t \in \Rows} g_t
        \sum_{\alpha \in \F_q^*} \innr{M_{t,\alpha}, x}^2 \right),
    \end{equation}
    where $\mathcal{G} := (g_t)_{t\in \Rows}$ is a sequence of
    independent standard Gaussian random variables.  Denote the term
    inside $\E_\Rows$ in \eqref{eqn:RIP:c} by $\cE_\Rows$; namely,
    \[
    \cE_{\Rows} := \E_{\mathcal{G}} \sup_{x \in \Ball} \left(\sum_{t
        \in \Rows} g_t \sum_{\alpha \in \F_q^*} \innr{M_{t,\alpha},
        x}^2 \right).
    \]

    Now we observe that, for any fixed $\Rows$, the quantity 
    $\cE_\Rows$ defines the supremum of a Gaussian process. 
    The Gaussian process
    $\{ G_x \}_{x \in \Ball}$ induces a pseudo-metric $\| \cdot \|_X$
    on $\Ball$ (and more generally, $\C^N$), where for $x, x' \in
    \Ball$, the (squared) distance is given by
    \begin{eqnarray}
      \| x - x' \|_X^2 &:=& \E_G |G_x - G_{x'}|^2 \nonumber \\
      &=& \sum_{t \in \Rows} \left( \sum_{\alpha \in \F_q^*} 
        \innr{M_{t,\alpha}, x}^2 - \sum_{\alpha \in \F_q^*} 
        \innr{M_{t,\alpha}, x'}^2 \right)^2 \nonumber \\
      &=& \sum_{t \in \Rows} \left( \sum_{\alpha \in \F_q^*} 
        \innr{M_{t,\alpha}, x+x'} \innr{M_{t,\alpha}, x-x'} 
      \right)^2. \label{eqn:RIP:d}
    \end{eqnarray}
    By Cauchy-Schwarz, \eqref{eqn:RIP:d} can be bounded as

\begin{eqnarray}
  \| x - x' \|_X^2 &\leq&
  \sum_{t \in \Rows} \left( \sum_{\alpha \in \F_q^*} 
    \innr{M_{t,\alpha}, x+x'}^2 \right) 
  \left( \sum_{\alpha \in \F_q^*} 
    \innr{M_{t,\alpha}, x-x'}^2 \right) \label{eqn:RIP:q}  \\
  &\leq& 
  \sum_{t \in \Rows} \sum_{\alpha \in \F_q^*} 
  \innr{M_{t,\alpha}, x+x'}^2
  \max_{t \in \Rows}
  \left( \sum_{\alpha \in \F_q^*} 
    \innr{M_{t,\alpha}, x-x'}^2 \right). \label{eqn:RIP:t}
\end{eqnarray}
Here is where our analysis differs from \cite{ref:RV08}.  When $q=2$,
\eqref{eqn:RIP:t} is exactly how the Gaussian metric is bounded in
\cite{ref:RV08}.  We obtain our improvement by bounding the metric in
a different way. Specifically, let $\eta \in (0,1]$ be a positive real
parameter to be determined later and let $r := 1+\eta$ and $s :=
1+1/\eta$ such that $1/r+1/s=1$. We assume that $\eta$ is so that $s$
becomes an integer.  We use H\"older's inequality with parameters $r$
and $s$ along with \eqref{eqn:RIP:q} to bound the metric as follows:
\begin{eqnarray}
  \| x - x' \|_X &\leq&
  \left(\sum_{t \in \Rows} \Big( \sum_{\alpha \in \F_q^*} 
    \innr{M_{t,\alpha}, x+x'}^2 \Big)^{r} \right)^{1/2r}
  \left(\sum_{t \in \Rows} \Big( \sum_{\alpha \in \F_q^*} 
    \innr{M_{t,\alpha}, x-x'}^2 \Big)^{s} \right)^{1/2s}.
  \label{eqn:RIP:metricBound}
\end{eqnarray}
Since $\|x\|_2=1$, $x$ is $k$-sparse, and $|M_{t,\alpha}|=1$ for all
choices of $(t,\alpha)$, Cauchy-Schwarz implies that
$\innr{M_{t,\alpha}, x}^2 \leq k$ and thus, using the triangle
inequality, we know that
\[ \sum_{\alpha \in \F_q^*} \innr{M_{t,\alpha}, x+x'}^2 \leq 4qk. \]
Therefore, for every $t \in \Rows$, seeing that $r = 1+\eta$, we have
\[
\Big( \sum_{\alpha \in \F_q^*} \innr{M_{t,\alpha}, x+x'}^2 \Big)^{r}
\leq (4qk)^{\eta} \sum_{\alpha \in \F_q^*} \innr{M_{t,\alpha},
  x+x'}^2,
\]
which, applied to the bound \eqref{eqn:RIP:metricBound} on the metric,
yields
% \footnote{The term $(4qk)^\eta$ on the left hand side of
% \eqref{eqn:RIP:e} should more precisely be
% $(4qk)^{\eta/2r}$. However, since $r \geq 1$, replacing the exponent
% $\eta/2r$ by $\eta$ is valid since it can only make the expression
% larger. The reader should note that we have implicitly made other
% simplifications of this kind at various other points in the proof.}
\begin{eqnarray}
  \| x - x' \|_X &\leq&  (4qk)^{\eta/2r}
  {\underbrace{\left( \sum_{t \in \Rows} \sum_{\alpha \in \F_q^*} 
        \innr{M_{t,\alpha}, x+x'}^2 \right)}_{\cE_2}}^{1/2r}
  \left(\sum_{t \in \Rows} \Big( \sum_{\alpha \in \F_q^*} 
    \innr{M_{t,\alpha}, x-x'}^2 \Big)^{s} \right)^{1/2s}. \label{eqn:RIP:e}
\end{eqnarray}
Now,
\begin{eqnarray}
  \cE_2 &\leq& 2 \left( \sum_{t \in \Rows} \sum_{\alpha \in \F_q^*} 
    \innr{M_{t,\alpha}, x}^2 + \sum_{t \in \Rows} \sum_{\alpha \in \F_q^*} 
    \innr{M_{t,\alpha}, x'}^2 \right) \leq 4 \cE'_\Rows,
\end{eqnarray}
where we have defined
\begin{equation}
  \cE'_\Rows := \sup_{x \in \Ball} \sum_{t \in \Rows} \sum_{\alpha \in
    \F_q^*} \innr{M_{t,\alpha}, x}^2 \label{eqn:RIP:y}.
\end{equation}
Observe that, by the triangle inequality,
\begin{equation}
  \cE'_\Rows \leq \sup_{x \in \Ball} \left|\sum_{t \in \Rows} \sum_{\alpha \in
      \F_q^*} \innr{M_{t,\alpha}, x}^2 -|\Rows|(q-1) \right| + |\Rows|(q-1). \label{eqn:RIP:z}
\end{equation}
Plugging \eqref{eqn:RIP:y} back in \eqref{eqn:RIP:e}, we so far have
\begin{equation}
  \| x - x' \|_X \leq
  4 (4qk)^{\eta/2r} {\cE'_{\Rows}}^{1/2r}
  \left(\sum_{t \in \Rows} \Big( \sum_{\alpha \in \F_q^*} 
    \innr{M_{t,\alpha}, x-x'}^2 \Big)^{s} \right)^{1/2s}. 
  \label{eqn:RIP:h}
\end{equation}

For a real parameter $u > 0$, define $N_X(u)$ as the minimum number of
spheres of radius $u$ required to cover $\Ball$ with respect to the
metric $\| \cdot \|_X$.  We can now apply Dudley's theorem on supremum
of Gaussian processes (cf.~\cite[Theorem~11.17]{ref:banach}) and
deduce that
\begin{equation}
  \cE_\Rows \lsim \int_{u=0}^{\infty} \sqrt{\log N_X(u)} du. \label{eqn:RIP:i}
\end{equation}

In order to make the metric $\| \cdot \|_X$ easier to work with, we
define a related metric $\| \cdot \|_{X'}$ on $\Ball$, according to
the right hand side of \eqref{eqn:RIP:h}, as follows:

\begin{equation}
  \| x - x' \|_{X'}^{2s} 
  := \sum_{t \in \Rows} \Big( \sum_{\alpha \in \F_q^*} 
  \innr{M_{t,\alpha}, x-x'}^2 \Big)^{s}.
  \label{eqn:RIP:g}
\end{equation}
Let $K$ denote the diameter of $\Ball$ under the metric $\|\cdot
\|_{X'}$. Trivially, $K \leq 2 |\Rows|^{1/2s} \sqrt{qk}$.  By
\eqref{eqn:RIP:h}, we know that
\begin{equation}
  \| x - x' \|_X \leq
  4 (4qk)^{\eta/2r} {\cE'_{\Rows}}^{1/2r}
  \| x - x' \|_{X'}. \label{rqn:RIP:h}
\end{equation}
Define $N_{X'}(u)$ similar to $N_X(u)$, but with respect to the new
metric $X'$.  The preceding upper bound \eqref{rqn:RIP:h} thus implies
that
\begin{equation}
  N_X(u) \leq N_{X'}(u/(4 (4qk)^{\eta/2r} {\cE'_{\Rows}}^{1/2r})).
\end{equation}
Now, using this bound in \eqref{eqn:RIP:i} and after a change of
variables, we have
\begin{equation}
  \cE_\Rows \lsim (4qk)^{\eta/2r} {\cE'_{\Rows}}^{1/2r} \int_{u=0}^{\infty} \sqrt{\log N_{X'}(u)} du. \label{eqn:RIP:p}
\end{equation}
Now we take an expectation over $\Rows$. Note that \eqref{eqn:RIP:z}
combined with \eqref{eqn:RIP:w} implies
\begin{equation} %
  \E_\Rows \cE'_\Rows \leq \cE + |\Rows|(q-1). \label{eqn:RIP:r} 
\end{equation}
Using \eqref{eqn:RIP:i}, we get
\begin{eqnarray*}
  \cE^{2r} &\stackrel{\eqref{eqn:RIP:c}}{\lsim}& \cE_1^{2r} = (\E_\Rows \cE_\Rows)^{2r} \leq \E_\Rows \cE_\Rows^{2r} \\ &\lsim&
  (4qk)^{\eta} \E_\Rows \left( (\cE'_\Rows)^{1/2r} \int_{u=0}^{\infty} \sqrt{\log N_{X'}(u)} du \right)^{2r} \\
  &\leq& (4qk)^{\eta} (\E_\Rows \cE'_\Rows) \max_{\Rows} \left( \int_{u=0}^{\infty} \sqrt{\log N_{X'}(u)} du \right)^{2r} \\
  &\stackrel{\eqref{eqn:RIP:r}}{\leq}& (4qk)^{\eta} (\cE+|\Rows|(q-1)) \max_{\Rows} \left( \int_{u=0}^{\infty} \sqrt{\log N_{X'}(u)} du \right)^{2r}.
\end{eqnarray*}
Define \newcommand{\bE}{\bar{\cE}}
\begin{equation} \label{eqn:RIP:x} \bE := \cE \cdot \left(
    \frac{\cE}{\cE+|\Rows|(q-1)} \right)^{1/(1+2\eta)}.
\end{equation}
Therefore, recalling that $r = 1+\eta$, the above inequality
simplifies to
\begin{equation}
  \bE \lsim (4qk)^{\eta} \max_{\Rows} \left( \int_{u=0}^{K} \sqrt{\log N_{X'}(u)} du 
  \right)^{1+1/(1+2\eta)},
  \label{eqn:RIP:j}
\end{equation}
where we have replaced the upper limit of integration by the diameter
of $\Ball$ under the metric $\|\cdot\|_{X'}$ (obviously, $N_{X'}(u) =
1$ for all $u \geq K$).

Now we estimate $N_{X'}(u)$ in two ways.  The first estimate is the
simple volumetric estimate (cf.~\cite{ref:RV08}) that gives
\begin{equation} \label{eqn:RIP:k} \log N_{X'}(u) \lsim k \log(N/k) +
  k \log(1+2K/u).
\end{equation}
This estimate is useful when $u$ is small.  For larger values of $u$,
we use a different estimate as follows.

\begin{claim} \label{clm:empirical} $\log N_{X'}(u) \lsim
  |\Rows|^{1/s} (\log N) q k s /u^2.$
\end{claim}

\begin{proof}
  We use the method used in \cite{ref:RV08} (originally attributed to
  B.~Maurey, cf.~\cite[\S~1]{ref:Carl85}) and empirically estimate any
  fixed real vector $x = (x_1, \ldots, x_N) \in \Ball$ by an
  $m$-sparse random vector $Z$, for sufficiently large $m$. The vector
  $Z$ is an average
  \begin{equation}
    \label{eq:def-of-Z}
    Z := \frac{\sqrt{k}}{m} \sum_{i = 1}^m Z_i,
  \end{equation}
  where each $Z_i$ is a $1$-sparse vector in $\C^N$ and $\E[Z_i] =
  x/\sqrt{k}$. The $Z_i$ are independent and identically distributed.

  The way each $Z_i$ is sampled is as follows. Let $x' := x/\sqrt{k}$
  so that $\|x'\|_1 = \frac{\|x\|_1}{\sqrt{k}} \leq 1$.  With
  probability $1-\|x'\|$, we set $Z_i := 0$. With the remaining
  probability, $Z_i$ is sampled by picking a random $j \in \supp(x)$
  according to the probabilities defined by absolute values of the
  entries of $x'$, and setting $Z_i = \sgn(x'_j) e_j$, where $e_j$ is
  the $j$th standard basis vector\footnote{Note that, since we have
    assumed $x$ is a real vector, $\sgn(\cdot)$ is always
    well-defined.}. This ensures that $\E[Z_i] = x'$.  Thus, by
  linearity of expectation, it is clear that $\E[Z] = x$.  Now,
  consider \[ \cE_3 := \E \| Z - x \|_{X'}. \] \Vnote{small change --
    changed $m$-sparse $Z$ to one of form \eqref{eq:def-of-Z}, and
    gave explicit bound on number of possible $Z$'s} If we pick $m$
  large enough to ensure that $\cE_3 \leq u$, regardless of the
  initial choice of $x$, then we can conclude that for every $x$,
  there exists a $Z$ of the form \eqref{eq:def-of-Z} that is at
  distance at most $u$ from $x$ (since there is always some fixing of
  the randomness that attains the expectation). In particular, the set
  of balls centered at all possible realizations of $Z$ would cover
  $\Ball$. Since the number of possible choices of $Z$ of the form
  \eqref{eq:def-of-Z} is at most $(2N+1)^m$, we have
  \begin{equation}
    \log N_{X'}(u) \lsim m \log N. \label{eqn:RIP:l}
  \end{equation}

  In order to estimate the number of independent samples $m$, we use
  symmetrization again to estimate the deviation of $Z$ from its
  expectation $x$. Namely, since the $Z_i$ are independent, by the
  symmetrization technique stated in Proposition~\ref{prop:sym} we
  have \Vnote{I think the $\|x\|_1$ here should be $\sqrt{k}$; made
    this change and propagated it.}
  \begin{equation} \label{eqn:RIP:v} %\E \| Z-x \|_{X'} \leq
    \cE_3 \lsim \frac{\sqrt{k}}{m} \cdot \E \left\| \sum_{i=1}^m
      \eps_i Z_i \right\|_{X'},
  \end{equation}
  where $(\eps_i)_{i \in [m]}$ is a sequence of independent Rademacher
  random variables in $\{-1,+1\}$.  Now, consider
  \begin{eqnarray}
    \cE_4 &:=& \E \left\| \sum_{i=1}^m \eps_i Z_i \right\|_{X'}^{2s} \nonumber \\ &=&
    \E \sum_{t \in \Rows} \Big( \sum_{\alpha \in \F_q^*} \innr{M_{t,\alpha}, \sum_{i=1}^m \eps_i Z_i }^2
    \Big)^s \nonumber  \\
    &=&
    \sum_{t \in \Rows} \E \left( \sum_{\alpha \in \F_q^*} \Big( \sum_{i=1}^m \eps_i \innr{M_{t,\alpha}, Z_i }
      \Big)^2 \right)^s \nonumber \\
    &=&
    \sum_{t \in \Rows} \E \left( \sum_{i,j=1}^m \eps_i \eps_j \sum_{\alpha \in \F_q^*} \innr{M_{t,\alpha}, Z_i } \innr{M_{t,\alpha}, Z_j }^\ast
    \right)^s. \label{eqn:RIP:s} 
  \end{eqnarray}
  Since the entries of the matrix $M$ are bounded in magnitude by $1$,
  we have
  \[
  \Big| \sum_{\alpha \in \F_q^*} \innr{M_{t,\alpha}, Z_i }
  \innr{M_{t,\alpha}, Z_j }^\ast \Big| \leq q.
  \]
  Using this bound and Proposition~\ref{prop:moment},
  \eqref{eqn:RIP:s} can be simplified as
  \begin{equation*}
    \cE_4 = \E \left\| \sum_{i=1}^m \eps_i Z_i \right\|_{X'}^{2s} \leq |\Rows| (4 q m s)^{s},
  \end{equation*}
  \Vnote{Changed two occurrences of $\|x\|_1$ to $\sqrt{k}$} and
  combined with \eqref{eqn:RIP:v}, and using Jensen's inequality,
  \[
  \cE_3 \lsim |\Rows|^{1/2s} \sqrt{ 4 q k s /m }.
  \]
  Therefore, we can ensure that $\cE_3 \leq u$, as desired, for some
  large enough choice of $m$; specifically, for some $m \lsim
  |\Rows|^{1/s} q k s /u^2$.  Now from \eqref{eqn:RIP:l}, we get
  \begin{equation} \label{eqn:RIP:m} \log N_{X'}(u) \lsim
    |\Rows|^{1/s} (\log N) q k s /u^2.
  \end{equation}
  Claim~\ref{clm:empirical} is now proved.
\end{proof}
Now we continue the proof of Lemma~\ref{lem:RIP:expectation}.  Break
the integration in \eqref{eqn:RIP:j} into two intervals. Consider
\[
\cE_5 := \underbrace{\int_{u=0}^{A} \sqrt{\log N_{X'}(u)} du}_{\cE_6}
+ \underbrace{\int_{u=A}^{K} \sqrt{\log N_{X'}(u)} du}_{\cE_7},
\]
where $A := K/\sqrt{qk}$.  We claim the following bound on $\cE_5$.

\begin{claim} \label{claim:dudley} $ \cE_5 \lsim |\Rows|^{1/2s} \sqrt{
    (\log N) q k s} \log(qk).$
\end{claim}

\begin{proof}
  First, we use \eqref{eqn:RIP:k} to bound $\cE_6$ as follows.
  \begin{equation} \label{eq:dud} \cE_6 \lsim A \sqrt{k \log(N/k)} +
    \sqrt{k} \int_{u=0}^{A} \sqrt{\ln(1+2K/u)} du.
  \end{equation}
  Observe that $2 K/u \geq 1$, so $1 + 2K/u \leq 4K/u$. Thus,
  \begin{eqnarray}
    \int_0^A \sqrt{\ln(1 + 2K/u)}\,du &\leq& \int_0^A \sqrt{\ln(4K/u)}\,du \nonumber \\
    &=& 2K \int_0^{A/2K} \sqrt{\ln(2/u)}\,du \nonumber \\
    &=& 2K \left(\frac{A}{2K}\sqrt{\ln (4K/A)} + \sqrt{\pi}\left(1 -
        \erf\left(\sqrt{\ln(4K/A)}\right)\right)\right) \nonumber \\
    &=& A\sqrt{\ln(4K/A)} + 2\sqrt{\pi} K\, \erfc\left(\sqrt{\ln(4K/A)}\right), \label{eqn:RIP:erf}
  \end{eqnarray}
  where $\erf(\cdot)$ is the Gaussian error function $\erf(x) :=
  \frac{2}{\sqrt{\pi}} \int_{t=0}^x e^{-t^2}dt$, and $\erfc(x) :=
  1-\erf(x)$, and we have used the integral identity
  \[
  \int \sqrt{\ln(1/x)} dx = -\frac{\sqrt{\pi}}{2}
  \erf\Big(\sqrt{\ln(1/x)}\Big) + x \sqrt{\ln(1/x)} + C
  \]
  that can be verified by taking derivatives of both sides.  Let us
  use the following upper bound
  \[
  (\forall x > 0) \quad \text{erfc}(x) =
  \frac{2}{\sqrt{\pi}}\int_{t=x}^\infty e^{-t^2}dt \leq
  \frac{2}{\sqrt{\pi}}\int_{t=x}^\infty \frac{t}{x} e^{-t^2}dt =
  \frac{1}{\sqrt{\pi}}\cdot \frac{e^{-x^2}}{x},
  \]
  and plug it into \eqref{eqn:RIP:erf} to obtain
  \begin{eqnarray*}
    \int_0^A \sqrt{\ln(1 + 2K/u)}\,du &\leq& A\sqrt{\ln(4K/A)} + 2\sqrt{\pi}
    K\left(\frac{1}{\sqrt{\pi}}\cdot\frac{A}{4K}\cdot \frac{1}{\sqrt{\ln(4K/A)}}\right)\\
    &=& A\sqrt{\ln(4K/A)} + \frac{A}{2\sqrt{\ln(4K/A)}} \\
    &\lsim& A\sqrt{\log(qk)} \lsim |\Rows|^{1/2s}\sqrt{\log(qk)}),
  \end{eqnarray*}
  where the last inequality holds since $A = K/\sqrt{qk} \lsim
  |\Rows|^{1/2s}$.  Therefore, by \eqref{eq:dud} we get
  \begin{equation} \label{eqn:RIP:n} \cE_6 \lsim
    |\Rows|^{1/2s}\sqrt{k}(\sqrt{\log N} + \sqrt{\log(qk)}).
  \end{equation}
  On the other hand, we use Claim~\ref{clm:empirical} to bound
  $\cE_7$.
  \begin{eqnarray}
    \cE_7 &\lsim& \sqrt{|\Rows|^{1/s} (\log N) q k s} \int_{u=A}^{K} du/u \nonumber \\
    &\lsim& |\Rows|^{1/2s} \sqrt{ (\log N) q k s} \log(qk). \label{eqn:RIP:o}
  \end{eqnarray}
  Combining \eqref{eqn:RIP:n} and \eqref{eqn:RIP:o}, we conclude that
  for every fixed $T$,
  \begin{equation*}
    \cE_5 = \cE_6 + \cE_7 \lsim |\Rows|^{1/2s} \sqrt{ (\log N) q k s} \log(qk).
  \end{equation*}
  Claim~\ref{claim:dudley} is now proved.
\end{proof}
By combining Claim~\ref{claim:dudley} and \eqref{eqn:RIP:j}, we have
\begin{eqnarray}
  \bE &\lsim&  
  (4qk)^{\eta} \max_{\Rows} \cE_5^{1+1/(1+2\eta)} \nonumber \\ &\lsim&
  (4qk)^{\eta} \Big(|\Rows|^{1/2s} \sqrt{ (\log N) q k s} \log(qk)\Big)^{1+1/(1+2\eta)} \nonumber \\
  &=&
  (4qk)^{\eta} |\Rows|^{\eta/(1+2\eta)} \Big(\sqrt{ (\log N) q k s} \log(qk)\Big)^{1+1/(1+2\eta)}.
  \label{eqn:RIP:u}
\end{eqnarray}
By Proposition~\ref{prop:deltaSqr} (setting $a := \cE/(|\Rows|(q-1))$
and $\mu := 2\eta$), and recalling the definition \eqref{eqn:RIP:x} of
$\bE$, in order to ensure that $\cE \leq \delta' (q-1) |\Rows|$, it
suffices to have
\begin{equation} \label{eqn:RIP:aa} \bE \leq
  {\delta'}^{\frac{2(1+\eta)}{1+2\eta}} |\Rows|(q-1)/4.
\end{equation}
Using \eqref{eqn:RIP:u}, and after simple manipulations,
\eqref{eqn:RIP:aa} can be ensured for some
\[
|\Rows| \lsim \frac{(4qk)^{2\eta}}{\eta} k (\log N) \log^2(qk) /
{\delta'}^2.
\]
This expression is minimized for some $\eta = 1/\Theta(\log(qk))$,
which gives
\[
|\Rows| \lsim k (\log N) \log^3(qk) / {\delta'}^2.
\]
This concludes the proof of Lemma~\ref{lem:RIP:expectation}.
\end{proof}

Now we turn to the tail bound on the random variable $\Delta$ and
estimate the appropriate size of $\Rows$ required to ensure that
$\Pr[\Delta > \delta |\Rows| (q-1)] \leq \gamma$.  We observe that the
tail bound proved in $\cite{ref:RV08}$ uses the bound on $\E[\Delta]$
as a black box. In particular, the following lemma, for $q=2$, is
implicit in the proof of Theorem~3.9 in $\cite{ref:RV08}$ (the
extension to arbitrary alphabet size $q$ requires only syntactical
modifications to the exact argument in $\cite{ref:RV08}$).

\begin{lem} \cite[implicit]{ref:RV08} Suppose that, for some $\delta'
  > 0$, $\E[\Delta] \leq \delta' |\Rows| (q-1)$. Then, there are
  absolute constants $c_1, c_2, c_3$ such that for every $\lambda >
  0$,
  \[
  \Pr[ \Delta > (c_1 + c_2 \lambda) \delta' |\Rows|(q-1)] \leq 3
  \exp(-\lambda^2),
  \]
  provided that \begin{gather} |\Rows|/k \geq c_3
    \sqrt{\lambda}/{\delta'}. \label{eqn:tailN} \end{gather} \qed
\end{lem}
Now it suffices to instantiate the above lemma with $\lambda :=
\sqrt{\ln(3/\gamma)}$ and $\delta' := \delta/(c_1 + c_2 \lambda) =
\delta/\Theta(\sqrt{\ln(3/\gamma)})$, and use the resulting value of
$\delta'$ in Lemma~\ref{lem:RIP:expectation}. Since
Lemma~\ref{lem:RIP:expectation} ensures that $|\Rows|/k = \Omega(\log
N)$, we can take $N$ large enough (depending on $\delta, \gamma$) so
that \eqref{eqn:tailN} is satisfied.  This completes the proof of
Theorem~\ref{thm:RIP}.
\end{proof}

% \begin{remark}
The proof of Theorem~\ref{thm:RIP} does not use any property of the
DFT-based matrix other than orthogonality and boundedness of the
entries. However, for syntactical reasons, that is, the way the matrix
is defined for $q>2$, we have presented the theorem and its proof for
the special case of the DFT-based matrices. The proof goes through
with no technical changes for any orthogonal matrix with bounded
entries (as is the case for the original proof of \cite{ref:RV08}).
In particular, we remark that the following variation of
Theorem~\ref{thm:RIP} also holds:

\begin{thm}
  Let $A \in \C^{N \times N}$ be any orthonormal matrix with entries
  bounded by $O(1/\sqrt{N})$.  Let $\Rows$ be a random multiset of
  rows of $A$, where $|\Rows|$ is fixed and each element of $\Rows$ is
  chosen uniformly at random, and independently with replacement.
  Then, for every $\delta, \gamma > 0$, and assuming $N \geq
  N_0(\delta, \gamma)$, with probability at least $1-\gamma$ the
  matrix $(\sqrt{N/|\Rows|}) A_\Rows$ satisfies RIP-2 of order $k$
  with constant $\delta$ for a choice of $|\Rows|$ satisfying
  \Vnote{small typo fix: $\log n \to \log N$}
  \begin{align*}
    |\Rows| \lsim \frac{\log(1/\gamma)}{\delta^2} k (\log N) \log^3 k.
    \tag*{\qedhere \qed}
  \end{align*}
\end{thm}
% \end{remark}

% ==============================================================================
% ==============================================================================

% \section{Concluding remarks}

% ==============================================================================
% ==============================================================================

\bibliographystyle{alpha} \bibliography{rip}
\appendix

\section{Useful tools}

The original definition of RIP-2 given in Definition~\ref{def:RIP}
considers all complex vectors $x \in \C^n$. Below we show that it
suffices to satisfy the property only for real-valued vectors $x$.

\begin{prop} \label{prop:RIPreal} Let $M \in \C^{m \times N}$ be a
  complex matrix such that $M^\dagger M \in \R^{N\times N}$ and for
  any $k$-sparse vector $x \in \R^N$, we have
  \[
  (1-\delta) \| x \|_2^2 \leq \| Mx \|_2^2 \leq (1+\delta) \| x
  \|_2^2.
  \]
  Then, $M$ satisfies RIP-2 of order $k$ with constant $\delta$.
\end{prop}

\begin{proof}
  Let $x = a + \I b$, for some $a, b \in \R^N$, be any complex vector.
  We have $\|x\|_2^2 = \|a\|_2^2+\|b\|_2^2$, and
  \begin{eqnarray*}
    \Big| \| Mx \|_2^2 - \|x\|_2^2 \Big| &=& \Big| x^\dagger M^\dagger M x - \|x\|_2^2 \Big| \\
    &=& \Big| (a^\dagger - \I b^\dagger) M^\dagger M (a + \I b) - \|x\|_2^2 \Big| \\
    &=& \Big| a^\dagger M^\dagger M a + b^\dagger M^\dagger M b + 
    \I (a^\dagger M^\dagger M b - b^\dagger M^\dagger M a)
    - \|x\|_2^2 \Big| \\
    &\stackrel{(\star)}{=}& \Big| a^\dagger M^\dagger M a + b^\dagger M^\dagger M b 
    - \|x\|_2^2 \Big| \\
    &=& \Big| a^\dagger M^\dagger M a - \|a\|_2^2+ b^\dagger M^\dagger M b - \|b\|_2^2\Big| \\
    &\stackrel{(\star\star)}{\leq}& \delta \|a\|_2^2 + \delta \|b\|_2^2 \\
    &=& \delta \|x\|_2^2,
  \end{eqnarray*}
  where $(\star)$ is due to the assumption that $M^\dagger M$ is real,
  which implies that $a^\dagger M^\dagger M b$ and $b^\dagger
  M^\dagger M a$ are conjugate real numbers (and thus, equal), and
  $(\star \star)$ is from the assumption that the RIP-2 condition is
  satisfied by $M$ for real-valued vectors and the triangle
  inequality.
\end{proof}

As a technical tool, we use the standard symmetrization technique
summarized in the following proposition for bounding deviation of
summation of independent random variables from the expectation. The
proof is a simple convexity argument (see, e.g.,
\cite[Lemma~6.3]{ref:banach} and \cite[Lemma~5.70]{ref:Ver12}).

\begin{prop} \label{prop:sym} Let $(X_i)_{i \in [m]}$ be a finite
  sequence of independent random variables in a Banach space, and
  $(\eps_i)_{i \in [m]}$ and $(g_i)_{i \in [m]}$ be sequences of
  independent Rademacher (i.e., each uniformly random in $\{-1,+1\}$)
  and standard Gaussian random variables, respectively. Then,
  \begin{align*}
    \E \Big\| \sum_{i \in [m]} (X_i - \E[X_i]) \Big\| \lsim \E \Big\|
    \sum_{i \in [m]} \eps_i X_i \Big\| \lsim \E \Big\| \sum_{i \in
      [m]} g_i X_i \Big\|.
  \end{align*}
  More generally, for a stochastic process $(X_i^{(\tau)})_{i \in [m],
    \tau \in \mathcal{T}}$ where $\mathcal{T}$ is an index set,
  \begin{align*}
    \E \sup_{\tau \in \mathcal{T}} \Big\| \sum_{i \in [m]}
    \big(X_i^{(\tau)} - \E[X_i^{(\tau)}]\big) \Big\| \lsim \E
    \sup_{\tau \in \mathcal{T}} \Big\| \sum_{i \in [m]} \eps_i
    X_i^{(\tau)} \Big\| \lsim \E \sup_{\tau \in \mathcal{T}} \Big\|
    \sum_{i \in [m]} g_i X_i^{(\tau)} \Big\|.
  \end{align*} \qed
\end{prop}

The following bound is used in the proof of Claim~\ref{clm:empirical},
a part of the proof of Lemma~\ref{lem:RIP:expectation}.

\begin{prop} \label{prop:moment} Let $(\eps_i)_{i \in [m]}$ be a
  sequence of independent Rademacher random variables, and
  $(a_{ij})_{i, j \in [m]}$ be a sequence of complex coefficients with
  magnitude bounded by $K$. Then,
  \[
  \left| \E \Big( \sum_{i,j \in [m]} a_{ij} \eps_i \eps_j \Big)^{s}
  \right| \leq (4Kms)^s.
  \]
\end{prop}

\begin{proof}
  By linearity of expectation, we can expand the moment as follows.
  \[
  \E \Big( \sum_{i,j \in [m]} a_{ij} \eps_i \eps_j \Big)^{s} =
  \sum_{\substack{(i_1, \ldots i_{s}) \in [m]^{s} \\
      (j_1, \ldots j_{s}) \in [m]^{s}}} \Big( a_{i_1 j_1} \cdots
  a_{i_s j_s} \E \Big[ \eps_{i_1} \cdots \eps_{i_s} \eps_{j_1} \cdots
  \eps_{j_s} \Big] \Big).
  \]
  Observe that $\E [\eps_{i_1} \cdots \eps_{i_s} \eps_{j_1} \cdots
  \eps_{j_s}]$ is equal to $1$ whenever all integers in the sequence
  \[(i_1, \ldots, i_s, j_1, \ldots, j_s)\] appear an even number of
  times.  Otherwise the expectation is zero. Denote by $S \subseteq
  [m]^{2s}$ the set of sequences $(i_1, \ldots, i_s, j_1, \ldots,
  j_s)$ that make the expectation non-zero.  Then,
  \[
  \left|\E \Big( \sum_{i,j \in [m]} a_{ij} \eps_i \eps_j
    \Big)^{s}\right| = \left|\sum_{(i_1, \ldots i_{s}, j_1, \ldots
      j_{s}) \in S} a_{i_1 j_1} \cdots a_{i_s j_s} \right| \leq K^{s}
  |S|.
  \]
  One way to generate a sequence $\sigma \in S$ is as follows. Pick
  $s$ coordinate positions of $\sigma$ out of the $2s$ available
  positions, fill out each position by an integer in $[m]$, duplicate
  each integer at an available unpicked slot (in some fixed order),
  and finally permute the $s$ positions of $\sigma$ that were not
  originally picked. Obviously, this procedure can generate every
  sequence in $S$ (although some sequences may be generated in many
  ways). The number of combinations that the combinatorial procedure
  can produce is bounded by $\binom{2s}{s} m^s (s!) \leq
  (4ms)^s$. Therefore, $|S| \leq (4ms)^s$ and the bound follows.
\end{proof}

We have used the following technical statement in the proof of
Lemma~\ref{lem:RIP:expectation}.

\begin{prop} \label{prop:deltaSqr} Suppose for real numbers $a > 0$,
  $\mu \in [0, 1]$, $\delta \in (0, 1]$, we have
  \[
  a \cdot \Big( \frac{a}{1+a} \Big)^{\frac{1}{1+\mu}} \leq
  \frac{\delta^{\frac{2+\mu}{1+\mu}}}{4}.
  \]
  Then, $a \leq \delta$.
\end{prop}

\begin{proof}
  Let $\delta' := \delta^{\frac{2+\mu}{1+\mu}} / 4^{\frac{1}{1+\mu}}
  \geq \delta^{\frac{2+\mu}{1+\mu}}/4$. From the assumption, we have
  \begin{equation} \label{eqn:deltaSqr:a} a \cdot \Big( \frac{a}{1+a}
    \Big)^{\frac{1}{1+\mu}} \leq \delta' \Rightarrow a^{2+\mu} \leq
    \delta^{2+\mu} (1+a)/4.
  \end{equation}
  Consider the function
  \[
  f(a) := a^{2+\mu} - \delta^{2+\mu} a/4 - \delta^{2+\mu}/4.
  \]
  The proof is complete if we show that, for every $a > 0$, the
  assumption $f(a) \leq 0$ implies $a \leq \delta$; or equivalently,
  $a > \delta \Rightarrow f(a) > 0$.  Note that $f(0) < 0$, and
  $f''(a) > 0$ for all $a > 0$. The function $f$ attains a negative
  value at zero and is convex at all points $a>0$. Therefore, it
  suffices to show that $f(\delta) > 0$. Now,
  \[
  f(\delta) = \delta^{2+\mu} - \delta^{3+\mu} /4 - \delta^{2+\mu}/4
  \geq (3 \delta^{2+\mu} - \delta^{3+\mu}) /4.
  \]
  Since $\delta \leq 1$, the last expression is positive, and the
  claim follows.
\end{proof}

\end{document}